\documentclass[a4paper, 12pt]{article}
\usepackage[T1]{fontenc}
\usepackage{lmodern}
\usepackage[margin=2cm]{geometry}
\usepackage{amsmath,amssymb,amsfonts}
\usepackage{fancyhdr}
\usepackage[utf8]{inputenc}  
\usepackage{color}
\usepackage{graphicx}
\usepackage{amsmath,amssymb,amsthm,mathrsfs,amsfonts,dsfont,bm} 
\usepackage{url}
\usepackage{authblk}
\usepackage{textcomp}
\usepackage[english]{babel}
\usepackage[autolanguage]{numprint}
\usepackage{graphicx}
\usepackage{verbatim}
\providecommand{\keywords}[1]{\textbf{\textit{key words : }} #1}

\newcounter{exocount}
\newcounter{questcount}

\renewcommand {\epsilon}{\varepsilon}

\renewcommand {\leq}{\leqslant}
\renewcommand {\geq}{\geqslant}

\theoremstyle{definition}
\theoremstyle{plain}

\newtheorem{exple}{Example}[section]
\newtheorem{thm}{Theorem}[section]

\newtheorem{lemma}[thm]{Lemma}

\title{$k$- price auctions and Combination-auctions}
\author{Martin Mihelich\thanks{martin.mihelich@walnut.ai} \ \ \   Yan Shu\thanks{ yan.shu@walnut.ai}}
\affil{Walnut Algorithms}

\begin{document}
\maketitle

\begin{abstract}
We provide  for the first time an exact analytical solution of the Nash equilibrium for $k$- price auctions. We also introduce a  generalization of the second price auction that maintains the incentive to bid truthfully, therefore paving the way for replacing second price auctions.
\end{abstract}

\keywords
Vickrey auctions, $k$- price auctions, Combination-auctions, game theory

\section{Introduction}

Second price auctions, also known as Vickrey auctions \cite{vickrey1961counterspeculation}, are well known and largely used as examples in online or Government auctions because it gives bidders an incentive to bid their true value. Nevertheless, second price auctions have reached some limits because, inter alia, bidders bid insincerely and the variance of the winning bid often has a large value. 

A natural generalization of second price auctions is the $k$-price auction, which has been studied by many researchers in the recent years. The reader can refer to \cite{PK99, KRISHNA20101,WILSON92} for related literature and \cite{EM04}. In particular, Roger B. Myerson~\cite{Roger1981} established the Revenue Equivalence Theorem (known as RET theorem) in 1981, which characterized the equilibrium strategy. Later in 1998, Monderer and Tenenholtz \cite{monderer2000k} proved the uniqueness of the  equilibrium strategies in $k$-price auctions for $k=3$. Under some regularity assumptions, they also provided sufficient conditions for the existence of the equilibrium. In 2000, Wolfsteller \cite{wolfstetter01third} solved the equilibrium $k$-price auctions for a uniform distribution and in 2014, Nawar and Sen \cite{nawar2014k} generalized the result of Wolfsteller, and provided a derivation expression of the $k$-price auctions bidding equilibrium for a quadratic distribution.

In this paper, we prove two major results. First, in Theorem~\ref{solution}, we generalized Wolfsteller and Nawar results by giving a closed form solution of the equilibrium of the $k$-price auctions in a general case. The solution takes an easy form which is easily calculable. As applications, this simplifies the discussions of the existence and uniqueness, and we can calculate the equilibrium for classical distributions and recover certain known results.

Secondly, we extend the notion of $k$- price auctions by introducing a new type of auction that we will call a "Combination-auctions":  the winner pays a centroid of the bids.  In fact a natural question is: Are there other combinations that second price auction which lead to a truthful equilibrium? We demonstrate in Theorem~\ref{combination1} that there exists combinations other than  the second price auction leading to a truthful equilibrium.These new strategies could replace second price auction... Then we  show in Theorem~\ref{combination2} that if there exists truthful strategies other than the second price auctions, then the  distribution of the valuations of the bidders is a uniform distribution.

We also provide an alternative proof of the RET theorem with probability tools, which does not involve advanced auctions theory and game theory knowledge.

This paper is organized as follows: in Section 2 we state the assumptions and re-establish the RET theorem with a probability approach. Next, we solve the equilibrium and discuss the uniqueness and the existence, together with applications for classical distributions. Finally, in Section 4 we introduce the "Combination-auctions" and discuss strategies other than second price auctions which are truthful.

\section{Problem formulation and assumptions}\label{S2}
In this section we present our assumptions and re-establish the RET theorem. The RET theorem is established in \cite{Roger1981} by Roger B. Myerson. The original proof involves heavy advanced auctions theory and game theory insight. Here we provide an alternative approach, with probability tools. With our approach, one can extend the RET theorem for "Combination-autions", which we will introduce in Section 4. 

\vskip 5mm

The $k$- price auction problem can be formulated as follows: consider a $k$-price auctions with $n$ bidders, where the highest bidder wins, and pays only the $k$-th highest bid. Let's assume that $k \geq 2$ and $n \leq k$. We make the assumptions following \cite{RS1981}:
\begin{enumerate}
    \item The valuations $V_i, i = 1,\cdots, n$ of the bidders are independent and identically distributed with distribution function $F$.
\item The distribution function $F$ is with values in $I$ where  $I= [0,1]$ or $I=\mathbb{R}^+$.
\end{enumerate}
We also assume that:
\begin{itemize}
\item[(A)]  $F$ is $k-2$ times continuously differentiable  and $ \forall x \in I$ $F'(x) >0$.
\end{itemize}
We remark that for analysis of the 3-price auction in the literature, the existence and the continuity of the density function are often assumed. It is thus natural to assume $(A)$ holds for the case of general k-price auctions.  
We also assume that each bidder bids $ X_i=g(V_i)$ where $g$ is a strictly increasing function. It follows that  the bids $X_i$ are independent variables and we denote their distribution function $\hat{F}$ and their density $\hat{f}$. As $g$ is strictly increasing, $\hat{F}$ is also strictly increasing on $g(I)$.

For $n-1$ bidders which bid $X_1,...,X_{n-1}$, we denote $Y_{n-1}= \max(X_1,...,X_{n-1})$ and $Y_{n-2}$ the second maximum. In general $\forall p \leq n-1$, $Y_{n-p}$ is the $p$-th maximum.

Now we recall the RET theorem for the k-price auctions (see also \cite{monderer2000k}):

\begin{thm}[RET theorem]
Let $k\geq 3$. A risk-neutral strategy $g$ is an equilibrium strategy in the k-price auction if and only if the following two conditions hold:
\begin{enumerate}
    \item $g$ is strictly increasing in the interval $[0,1]$.
    \item It holds for all $x\in [0,1]$:
    $$\int_{t = 0}^x (x-g(t))F(t)^{n-k}(F(x)-F(t))^{k-3}F'(t)dt = 0.$$
\end{enumerate}
\end{thm}

\subsection{An alternative approach of RET theorem}
This section provides a probability approach of RET theorem. Before proving the theorem, we first show several lemmas with algebra computation.

\begin{lemma}
For all $t\leq x$ and $(t,x)\in I^2$, denote 
$$H(t):= \mathbb{P}([Y_{n-p+1} \leq t] \cap [Y_{n-1} \leq x]).$$
Then it  holds
\begin{equation}\label{Ht}
   H(t) =  \sum_{p=0}^{k-2} \hat{F}(t)^{n-1-p} (\hat{F}(x)-\hat{F}(t))^p
\end{equation}
and 
\begin{equation}\label{H't}
 H'(t)=\frac{(n-1)!}{ (n-k)! (k-2)!} \hat{F}^{n-k}(t) (\hat{F}(x)-\hat{F}(t))^{k-2}   
\end{equation}
\end{lemma}

\begin{proof}
It holds for all
$t \leq x \text{   } (t,x) \in I^2 $:
\begin{align*}
 P(Y_{n-p+1}  \leq t \cap Y_{n-1} \leq x)=&  P(X_1,...,X_{n-1} \leq t)\\ &+ \binom{n-1}{1} P(X_1,...,X_{n-2} \leq t \cap X_{n-1} \in ]t,x]) \\ &+ ...+
\binom{n-1}{p-2} P(X_1,...,X_{n-p+1} \leq t \cap X_{n-p+2},...,X_{n-1} \in ]t,x])
\end{align*}

And equation \eqref{Ht} follows according to the fact that $X_i$ are independently identically distributed.

Now we prove equation \eqref{H't}.
Derive equation \eqref{Ht}, it follows:
\begin{equation}
\begin{split}
 H'(t)= \hat{f}(t) \left( \sum_{p=0}^{k-2} \binom{n-1}{p} (n-1-p)\hat{F}(t)^{n-2-p}(\hat{F}(x)-\hat{F}(t))^p \right. \\-  \sum_{p=1}^{k-2} \left. \binom{n-1}{p} p \hat{F}(t)^{n-1-p}(\hat{F}(x)-\hat{F}(t))^{p-1} \right)
\end{split}
\end{equation}


Then

\begin{equation}
\begin{split}
 H'(t)= \hat{f}(t) \left( \sum_{p=0}^{k-2} \binom{n-1}{n-1-p} (n-1-p) \hat{F}(t)^{n-2-p}(\hat{F}(x)-\hat{F}(t))^p \right. \\ - \sum_{p=0}^{k-3} \left. \binom{n-1}{p+1} (p+1) \hat{F}(t)^{n-2-p}(\hat{F}(x)-\hat{F}(t))^{p}\right)
\end{split}
\end{equation}


A telescopic sum  appears and after simplification we get:
$$H'(t)=\frac{(n-1)!}{ (n-k)! (k-2)!} \hat{F}^{n-k}(t) (\hat{F}(x)-\hat{F}(t))^{k-2} \hat{f}(t),$$
which is exactly equation \eqref{H't}.
\end{proof}

\begin{lemma} \label{integration}
For all $p,m \in \mathbb{N}$, it holds:
\begin{equation}\label{IPP}
\int_0^1 (1-u)^{p}u^{m} du = \frac{p!m!}{(p+m+1)!}. 
\end{equation}
For $n>k\geq 2$, it holds:
\begin{equation}\label{sum0}
\sum_{p=0}^{k-2} \frac{(-1)^{k-2-p}}{n-1-p}\binom{k-2}{p} = \frac{(n-k)!(k-2)!}{(n-1)!}.
\end{equation}
\end{lemma}

\begin{proof}
The equation \eqref{IPP} is direct from integration by parts. We now prove equation \eqref{sum0}. 
for $n>k\geq 2$, denote
$$A = \sum_{p=0}^{k-2} \frac{(-1)^{k-2-p}}{n-1-p}\binom{k-2}{p}$$
Consider 
$$r(x) = \sum_{p=0}^{k-2} (-x)^{n-2-p}\binom{k-2}{p} = (-1)^{n-k}x^{n-k}(1-x)^{k-2}. $$
Now let
$$R(x) = \int_{0}^{x} r(u)du.$$
Observe that
$$R(1) = (-1)^{n-k+1}A$$
and according to equation \eqref{IPP}, it holds
$$R(1) =  (-1)^{n-k+1} \frac{(n-k)!(k-2)!}{(n-1)!}$$
hence 
$$A = \frac{(n-k)!(k-2)!}{(n-1)!} = \frac{1}{(n-1)\binom{n-2}{k-2}}.$$
\end{proof}

Now we are ready to prove the RET Theorem.
\begin{proof}
We consider here $n-1$ bidders which play with the same rule $X_i=g(V_i)$. The  payoff expression of the $n$-th bidder for a valuation $v$ and a bid $x$ is:

\begin{equation}
 U(x,v)= \int_{0}^{x} (v-t) H'(t) dt=  \int_{0}^{x} (v-t) \frac{(n-1)!}{ (n-k)! (k-2)!} \hat{F}^{n-k}(t) (\hat{F}(x)-\hat{F}(t))^{k-2} \hat{f}(t) dt
\label{bid1}
\end{equation}

The goal is at $v$ fixed to find $x=g(v)$ which maximizes the payoff.

\vskip 5mm


By developing \eqref{bid1} we find: 

$$ U(x,v)= \frac{(n-1)!}{ (n-k)! (k-2)!} \sum_{p=0}^{k-2} (-1)^{k-2-p} \binom{k-2}{p} \hat{F}^p(x) \int_0^x (v-t) \hat{F}^{n-2-p}(t) \hat{f}(t) dt $$

And then  $$\frac{ \partial U}{\partial x}(x,v) =0$$

is equivalent to
\begin{equation}\label{Upartialx}
 \sum_{p=0}^{k-2} \frac{(-1)^{k-2-p}}{n-1-p} \binom{k-2}{p} \left( (n-1)(v-x)\hat{F}^{n-2}(x)+p\hat{F}^{p-1}(x) \int_0^x \hat{F}^{n-1-p}(t) dt \right)=0.
\end{equation}


Now let $g_1$ be the quantile function of the distribution $X$, i.e. $g_1 = \hat{F}^{-1}$. As $\hat{F}$ is strictly increasing, $g_1$ exists and differentiable almost everywhere. Let denote $\hat F(x) = a$ and $q$ the quantile function of $F$. Thus,  $X=g(V)$ implies that $a = \hat{F}(x)=F(v)$, therefore we have $v = q(a)$. Equality \eqref{Upartialx} becomes
\begin{equation}\label{Upartialxq}
 \sum_{p=0}^{k-2} \frac{(-1)^{k-2-p}}{n-1-p} \binom{k-2}{p} \left( (n-1)(q(a)-g_1(a))a^{n-2}+pa^{p-1} \int_0^a u^{n-1-p} g_1'(u)du \right)=0.
\end{equation}
Observe that
\begin{align*}
     &\sum_{p=0}^{k-2} \frac{(-1)^{k-2}}{n-1-p}\binom{k-2}{p} p a^{p-1}\int_0^a u^{n-1-p}g_1'(u)\ du \\
     = &\int_0^a \left(\sum_{p=0}^{k-2} \frac{(-1)^{k-2}}{n-1-p}\binom{k-2}{p} p a^{p-1} u^{n-1-p}\right)g_1'(u)\ du \\
     = & \int_0^a P(a,u)g_1'(u)\ du,
\end{align*}
where
$$
P(a,u) = \sum_{p=0}^{k-2} \frac{(-1)^{k-2}}{n-1-p}\binom{k-2}{p} p a^{p-1} u^{n-1-p}.
$$

Since, $n\geq k$, we deduce that $n-1-p>0$ for $p\in [0,k-2]$. Thus for all $a\in \mathbb{R}$, $P(a,0) = 0$.

Denote $\partial_a$ the derivation operator with respect to $a$ and $\partial_u$ the derivation operator with respect to $u$. It holds: 

\begin{align*}
\partial_u P(a,u) &= \sum_{p=0}^{k-2} (-1)^{k-2}\binom{k-2}{p} p a^{p-1} u^{n-2-p}\\
& = \left(\sum_{p=0}^{k-2} (-1)^{k-2}\binom{k-2}{p} p a^{p-1} u^{k-2-p}\right) u^{n-k}\\
& = \partial_a\left(\sum_{p=0}^{k-2} (-1)^{k-2}\binom{k-2}{p} a^{p} u^{k-2-p}\right) u^{n-k}\\
& = \partial_a \left((a-u)^{k-2}\right) u^{n-k}\\
& = (k-2)(a-u)^{k-3}u^{n-k}.
\end{align*}
Therefore, applying integration by parts and observing that $P(0,0) = 0$, it follows that 

\begin{align*}
    \int_0^a P(a,u)g_1'(u)du &= \left[P(a,u) g_1(u)\right]_0^a - \int_0^a \partial_u P(a,u) g_1(u) du\\
    & = P(a,a)g_1(a) -\int_0^a (p-2)(a-u)^{k-3}u^{n-k}g_1(u)du.
\end{align*}
According to lemma \ref{integration}
\begin{align*}
    P(a,a) &= \int_0^a (k-2)(a-u)^{k-3}u^{n-k}\ du\\
    & = (k-2)a^{n-2}\int_0^1 (1-u)^{k-3}u^{n-k}du \\
    & = a^{n-2} \frac{(n-k)!(k-2)!}{(n-2)!}. 
\end{align*}
The latter equality together with equation \eqref{IPP} leads to 

\begin{equation}
    (q(a)-g_1(a)) a^{n-2} \frac{(n-k)!(k-2)!}{(n-2)!} + a^{n-2} \frac{(n-k)!(k-2)!}{(n-2)!} g_1(a)  = \int_0^a (k-2)(a-u)^{k-3}u^{n-k}g_1(u)du
\end{equation}
which is equivalent to 
\begin{equation} \label{eqint}
    q(a) a^{n-2} \frac{(n-k)!(k-2)!}{(n-2)!}= \int_0^a (k-2)(a-u)^{k-3}u^{n-k}g_1(u)du.
\end{equation}
Notice that the latter equation can be also written as:
$$q(a) P(a,a)= \int_0^a (k-2)(a-u)^{k-3}u^{n-k}g_1(u)du,$$
rearranging the terms, it holds:
$$\int_0^a (k-2)(a-u)^{k-3}u^{n-k}(g_1(u)-q(a))du = 0.$$
Let $x = q(a)$ and change the variable $u = F(t)$ inside the integration, together with the fact that $g_1(u) = \hat{F}^{-1}\circ F(t) = g(t)$, we obtain:
$$\int_0^x (k-2)(F(x)-F(t))^{k-3}F(t)^{n-k}(g(t)-x)F'(t)dt = 0.$$
The proof is completed.

\end{proof}

\section{Analysis of equilibrium}
In this section,  we first give a closed form solution of the equilibrium for $k$- price auctions, $k\geq 3$. Then with the solution, we analyze the existence and uniqueness of the equilibrium. At the end of this section we provide some examples.

\subsection{Solution of the equilibrium}
\begin{thm}\label{solution}
We assume that $(A)$ holds, then the equilibrium satisfies equation \eqref{eqint}.  Moreover, equation \eqref{eqint} has unique solution:
\begin{equation}
    g_1(a) = \frac{a^{k-n}(n-k)!}{(n-2)!}\left(q(a) a^{n-2}\right)^{(k-2)} =\frac{\left(q(a) a^{n-2}\right)^{(k-2)}}{(a^{n-2})^{(k-2)}}.
\end{equation}
\end{thm}

\begin{proof}
We denote for $p\geq 0$
$$A_p(a,u) = (a-u)^p u^{n-k}g_1(u) $$
and 
$$G_p(a,u) = \int_0^a (a-u)^p u^{n-k}g_1(u). $$
It holds : 
$$A_0(a,a) = a^{n-k}g_1(a),$$
$$G_0(a,u) = \int_0^a u^{n-k}g_1(u) = \int_0^a A_0(a,u)du, $$
and for $p\geq 1$:
$$A_p(a,a) = 0,$$

For $p>1$ and for all $a\in [0,1]$:
\begin{align*}
    G_p'(a) &= \lim_{h\rightarrow 0} \frac{1}{h}\left(\int_0^{a+h} A_p(a+h,u)du - \int_0^a A_p(a,u) du\right)\\
    &= \lim_{h\rightarrow 0} \frac{1}{h}\left(\int_0^{a+h} A_p(a+h,u) du - \int_0^a A_p(a+h,u)du + \int_0^{a} A_p(a+h,u) - A_p(a,u) du \right)\\
    & = A_p(a,a)+ \lim_{h\rightarrow 0} \frac{1}{h} \left(\int_0^{a} A_p(a+h,u) - A_p(a,u) du\right)\\
    & = A_p(a,a)+ \int_0^a \partial_a A_p(a,u) du\\
    & = (p-1) G_{p-1}(a) 
\end{align*}

We remark that for all $b \in I$, $\int_{0}^{b} |g_1(u)| du \leq b g_1(b)$ holds, thus we can switch the integral and the limit in the previous formula. Therefore calculation is valid.

Therefore, the $(p+1)^{th}$ derivative of $G_p$ on $a$ is 
\begin{equation}
    G_p^{(p+1)}(a) = (p-1)!G_0'(a) = (p-1)! A_0(a,a) = (p-1)!a^{n-k}g_1(a).
\end{equation}
Hence, the $(k-2)^{th}$ derivative of equation\eqref{eqint} gives
\begin{align*}
 \left(q(a) a^{n-2} \frac{(n-k)!(k-2)!}{(n-2)!}\right)^{(k-2)}& = (k-2)\left(\int_0^a (a-u)^{k-3}u^{n-k}g_1(u)du\right)^{(k-2)}\\
 & = (k-2)G_{k-3}^{(k-2)}(a)\\
& = (k-2)!a^{n-k}g_1(a).
\end{align*}
Rearranging the terms, we have
$$
g_1(a) = \frac{a^{k-n}(n-k)!}{(n-2)!}\left(q(a) a^{n-2}\right)^{(k-2)}=\frac{\left(q(a) a^{n-2}\right)^{(k-2)}}{(a^{n-2})^{(k-2)}}.
$$
\end{proof}

Since equation \eqref{eqint} has a unique solution, according to RET theorem, the equilibrium exists if and only if $\left(q(a) a^{n-2}\right)^{(k-2)}/(a^{n-2})^{(k-2)}$ is strictly increasing.

\subsection{Examples}
As applications, we introduce several examples. With the formula from previous section, one can easily recover some classical results. 
\begin{exple}[Uniform distribution]
$q(a)= a, g(v) = g_1(a)$.
\begin{align*}
    &g(v) = g_1(a) =\frac{a^{k-n}(n-k)!}{(n-2)!}\left(q(a) a^{n-2}\right)^{(k-2)} \\
    = &\frac{a^{k-n}(n-k)!}{(n-2)!}\left( a^{n-1}\right)^{(k-2)}\\
    = & \frac{n-1}{n-k+1} a=  \frac{n-1}{n-k+1} v
\end{align*}
\end{exple}
\begin{exple}[3-price auctions]
For $k=3$, notice that
$q'(a) = {F^{-1}}'(a) = 1/F'(v)$,
it follows that:
\begin{align*}
    g(v) &= g_1(a) =\frac{a^{3-n}}{(n-2)}\left(q(a) a^{n-2}\right)' \\
    &= q(a)+\frac{1}{n-2}a q'(a)\\
    &= v+\frac{1}{n-2}\frac{F(v)}{F'(v)}
\end{align*}
and we found the well known result.
\end{exple}

\begin{exple}[ 4-price auctions].
For $k = 4$, since 
$$q'(a) = \frac{1}{F'(v)} = \frac{1}{F'(q(a))},$$
it follows that 
$$
q''(a) = -\frac{F''(q(a))}{F'^2(q(a)}q'(a) = -\frac{F''(v)}{F'^3(v)}.
$$ 
Then it holds : 
\begin{align*}
    g_1(a) &=\frac{a^{4-n}}{(n-2)(n-3)}\left(q(a) a^{n-2}\right)''\\&=\frac{a^{4-n}}{(n-2)(n-3)}\left((n-2)(n-3)a^{n-4}q(a) + 2(n-2)a^{n-3}q'(a)+a^{n-2}q''(a)\right) \\
    &= q(a) + \frac{2}{n-3}\frac{a}{F'(q(a))}-\frac{1}{(n-2)(n-3)}\frac{a^2F''(q(a))}{F'(q(a))^3}
\end{align*}
Finally, with with $v=q(a)$ and $a=F(v)$: 
    $$g(v)= v+\frac{2}{n-3}\frac{F(v)}{F'(v)}-\frac{1}{(n-2)(n-3)}\frac{F^2(v)F''(v)}{F'(v)^3}$$
We recover the result of \cite{nawar2014k}, Theorem 2.                                                                             

\end{exple}
\begin{exple}[polynomial distribution].
For $F(x):= x^\alpha$ with $\alpha >0$, then $q(a) = a^{1/\alpha}$ and 
\begin{align*}
    g(v) &= g_1(a) =\frac{a^{k-n}(n-k)!}{(n-2)!}\left(q(a) a^{n-2}\right)^{(k-2)} \\
    &= \frac{a^{k-n}(n-k)!}{(n-2)!}\left( a^{n-2+1/\alpha}\right)^{(k-2)}\\
    &=  \frac{\Gamma(n-k+1) \Gamma(n-1+1/\alpha)}{\Gamma(n-k+1+1/\alpha)\Gamma(n-1)} a^{1/\alpha}\\
    &= \frac{\Gamma(n-k+1) \Gamma(n-1+1/\alpha)}{\Gamma(n-k+1+1/\alpha)\Gamma(n-1)} v
\end{align*}

where $\Gamma$ is the Gamma function.
In particular, if $\alpha = \frac{1}{m}$ where $m$ is a positive integer, 
$$g(v) = \frac{(n-2+m)...(n-k+m+1)}{(n-2)...(n-k+1)}v$$
\end{exple}

\section{ Combination-price auctions}

In this section we start by defining a new type of auction that we call a Combination-price auction. We demonstrate that a Combination-price auction could be truthful for linear combinations other than second price auctions ($\alpha_2=1$), giving some companies an incentive to use it instead of second price auctions. Moreover, we will characterize the distributions for which there exists a linear combination different from $\alpha_2=1$ which is truthful.


\subsection{ Nash equilibrium of Combination-price auctions}

We  call Combination-price auctions an auction in which the winner will pay a linear combination of the prices bid by the bidders: $\alpha_1 Y_{n}+...+\alpha_s Y_{n-s+1}$ where all the $\alpha_k$ are positive satisfying $\sum_{k=1}^s \alpha_k=1$.

Like in the first part, we consider here $n-1$ bidders which play with the same rule $X_i=g(V_i)$. Reasoning as in Section \ref{S2},  the payoff expression for the $n$-th bidder for a valuation $v$ and a bid $x$ can be expressed as a multiple integral:
$$ U(x,v)=\iint_{0}^x (x-\sum_{k=1}^s \alpha_k t_k) P(Y_{n-1} \in [t_1,t_1+dt_1] \cap...\cap Y_{n-s+1} \in [t_s,t_s+dt_s] \cap Y_{n-1} \leq x)dt_1...dt_s.$$

Together with $\sum_{k=1}^s \alpha_k=1$ we can write:
$$ U(x,v)=\sum_{k=1}^s   \alpha_k \iint_{0}^x (x- t_k) P(Y_{n-1} \in [t_1,t_1+dt_1] \cap...\cap Y_{n-s+1} \in [t_s,t_s+dt_s] \cap Y_{n-1} \leq x)dt_1...dt_s$$

 then:
$$ U(x,v)= \sum_{k=1}^s \alpha_k U_k(x,v)$$ where $U_k$ is the payoff of the $k$-price auctions.

\vskip 3mm

Now combining equation \eqref{eqint}  with a simple calculation of $U_1$ and $U_2$,  we deduce that $g_1$ is a Nash equilibrium for the combination-price auctions if and only if $g_1$ is increasing and the following equation holds:
\begin{align}\label{combauct}
    &q(a) a^{n-2} \left( \sum_{k=2}^s \frac{(n-k)!(k-2)!}{(n-2)!} \alpha_p\right)-\alpha_1 a^{n-1} \nonumber
    \\ =&
\alpha_1 (n-1) \frac{q(a)-g_1(a)}{g_1'(a)} a^{n-2}+\alpha_2 a^{n-2} g_1(a)+\int_0^a g_1(u) \left(\sum_{k=3}^s \alpha_k (k-2)(a-u)^{k-3}u^{n-k}\right)du.
\end{align}

\subsection{ Study of truthful equilibrium}
The following theorem characterizes the truthful equilibrium for uniform distribution. 
\begin{thm}\label{combination1}
If the distribution is uniform, there are an infinite number of linear combinations which lead to a truthful strategy.

More precisely, adapting the notations before, a Combination-price auction is truthful if and only if the coefficients $(\alpha_i)$ satisfy the following conditions:
\begin{enumerate}
\item It holds $$ \alpha_1=\sum_{k=3}^s\frac{(k-3)!(n-k)!(2k-n-3)}{(n-2)!}\alpha_k,$$
\item for all $k$, $\alpha_k \geq 0$, and $$\sum_{k=1}^s \alpha_k=1.$$  
\end{enumerate}

\end{thm}

\begin{proof}

The equilibrium is truthful if and only if $ \forall \, a  \,  g_1(a)=q(a)$ and together with equation \eqref{combauct}, it holds:

\begin{equation}
\begin{split}
q(a) a^{n-2} \left( \sum_{k=3}^s \frac{(n-k)!(k-2)!}{(n-2)!} \alpha_k\right)-\alpha_1 a^{n-1}=\\
\int_0^a q(u) \left(\sum_{k=3}^s \alpha_k (k-2)(a-u)^{k-3}u^{n-k}\right)d u.
\end{split}
\label{fair1}
\end{equation}

\vskip 5mm

If the distribution is uniform in $[0,1]$, then $q(a)=a$ and by replacing the latter formula becomes
$$ \alpha_1=\sum_{k=3}^s \frac{(k-3)!(n-k)!(2k-n-3)}{(n-2)!}\alpha_k.$$ 

Moreover we still have $\sum_{k=1}^s \alpha_k=1$ and $ \forall \, k$  $\alpha_k \geq 0$.

\vskip 3mm

We can easily show that this intersection of two hyper plans with the portion of the space of $\alpha_k \geq 0$ has an infinite number of solutions. 
\end{proof}

\vskip 5mm
Next theorem characterizes the truthful equilibrium for general distribution.
\begin{thm}\label{combination2}
If $F$ is a continuous and strictly increasing function on $I=[0,1]$ or $I=[0, +\infty)$ and if it exists a truthful equilibrium different of second price auctions, then $F$ corresponds to the uniform distribution in $[0,1]$, i.e $F(x)=x$ for $x\in [0,1]$.
\end{thm}

\begin{proof}

As $F$ is continuous on $I$ and strictly increasing, then $q$ exists and is continuous on $J=[0,1]$ or $J=[0,1)$. Moreover, we have $q(0)=0$.

Then, with equation \eqref{fair1} we can easily deduce that $q$ is infinitely differentiable on $J \setminus \{0\}$. Moreover by differentiating  \eqref{fair1}  we can show that $q'$ is extendable by continuity in $0$ and then by differentiating an other time that $q''$ is also extendable by continuity in $0$. 

\vskip 5mm

With the variable change $u=av$ and after simplifying by $a^{n-2}$, equation \eqref{fair1} becomes:

\begin{equation}
\begin{split}
q(a) \left( \sum_{p=3}^s \frac{(n-p)!(p-2)!}{(n-2)!} \alpha_p\right)-\alpha_1 a=\\
\int_0^1 q(av) \left(\sum_{p=3}^s \alpha_p (p-2)(1-v)^{p-3}v^{n-p}\right)dv
\end{split}
\label{fair2}
\end{equation}

Then, by deriving the equation \eqref{fair2} two times (the derivation under the integral sum is valid because $q$ is two times continuously differentiable) we find $\forall \, a \in I$:
\begin{equation}
q''(a) \left( \sum_{p=3}^s \frac{(n-p)!(p-2)!}{(n-2)!} \alpha_p\right)= \int_0^1 v^2 q''(av) \left(\sum_{p=3}^s \alpha_p (p-2)(1-v)^{p-3}v^{n-p}\right)dv.
\label{fair3}
\end{equation}

Moreover, according to lemma \ref{IPP}, we have: 
$$\int_0^1 \sum_{p=3}^s \alpha_p (p-2)(1-v)^{p-3}v^{n-p}dv= \sum_{p=3}^s \frac{(n-p)!(p-2)!}{(n-2)!}.$$


Supposing by the absurd that it exists $ x \in I$ such that $q''(x) \neq 0$, and without loss of generality $q''(x) >0$. By noting $a_0$ such that $q(a_0)= \max_{I} q(x)$ (or $\max_{[0,b]}  q(x)$ $\forall \, b \in (0,1) $ if $I=[0,1)$) we have that the right hand side of \eqref{fair3} satisfies:
$$
\int_0^1 v^2 q''(a_0v) \left(\sum_{p=3}^s \alpha_p (p-2)(1-v)^{p-3}v^{n-p}\right)dv < q''(a_0) \left( \sum_{p=3}^s \frac{(n-p)!(p-2)!}{(n-2)!} \alpha_p\right),
$$ 
which is absurd because there is equality in \eqref{fair2}.

\vskip 3mm

Then $\forall \, a \in I, \,  q''(a)=0$ and we deduce that we necessarily have $I=[0,1]$ and such that $q(0)=0$ and $q(1)=1$ we have $q(a)=a$ which corresponds to an uniform distribution of the valuation of the bidders.
\end{proof}

\section{Conclusion}

In this paper we give an exact  analytical solution of the $k$-price auctions. Moreover we introduce a new type of auction which could replace the second price auctions. This article has led us to interesting open questions that we list here:

\begin{enumerate}
\item Is it possible to find the exact analytical solution of the combination-auctions, i.e solve the equation \eqref{eqint}?
\item Here we took constant coefficients $\alpha_i$. What happens if the $\alpha_i$ becomes random variables dependent, for example, on $F$?
\item For a given $F$ is it possible to find a combination such that $\alpha_2=0$ (or $\alpha_2\leq  r, \, r \in [0,1)$)  and such that the Nash equilibrium of these combination-auctions noted $g$ minimizes any standard norm between $g$ and the truthful one $x \mapsto x$?
\end{enumerate}


\end{document}